\documentclass[a4paper,11pt]{amsart}          
\usepackage{amsfonts,amsmath,latexsym,amssymb} % these packages are required
\usepackage[dvipsnames]{xcolor}          
\usepackage{float}
\usepackage{tikz}
\usepackage{circuitikz}
\usepackage{hyperref}
\usepackage{mathtools}
\usepackage{enumitem}
\usepackage{verbatim}
\usepackage{mathabx} 
\allowdisplaybreaks

\newtheorem{theorem}{Theorem}    
\numberwithin{theorem}{section}       
\newtheorem{lemma}[theorem]{Lemma}               
\newtheorem{corollary}[theorem]{Corollary}
\newtheorem{proposition}[theorem]{Proposition}

\theoremstyle{definition}

\newtheorem{definition}[theorem]{Definition}
\newtheorem{example}[theorem]{Example}

\newtheorem{remark}[theorem]{Remark}

\newtheorem{question}[theorem]{Question}

\newcommand{\AND}{\quad \text{and} \quad}
\newcommand{\C}{\mathbb C}%{\widetilde{\rho}}
\renewcommand{\Im}{\operatorname{\sf Im}}%imaginary part
\newcommand{\N}{\mathbb{N}}
\newcommand{\WWs}{\widehat\Omega}
\renewcommand{\P}{\mathcal{Y}}
\newcommand{\Ppi}{\boldsymbol{\Pi}}
\newcommand{\pq}{\textbf{\textit{p}}}
\newcommand{\fq}{\textbf{\textit{f}}}
\newcommand{\PQ}{\textbf{\textit{P}}\!}
\newcommand{\QQ}{\textbf{\textit{Q}}\!}
\newcommand{\R}{\mathbb R}%{\widetilde{\rho}}
\renewcommand{\Re}{\operatorname{\sf Re}}%real part

\newcommand{\tb}{\mathbf{t}}
\newcommand\uno{\mathbf{1}}
\newcommand{\wc}{\widecheck}
\newcommand{\wt}{\widetilde}
\newcommand{\ww}{\omega}%\mathsf{w}}
\newcommand{\WW}{\Omega}%\mathsf{W}}

\author[A. Muranova]{Anna Muranova}
\address{Faculty of Mathematics and Computer Science, University of Warmia 
and Mazury, ul. Sloneczna 54, 10-710 Olsztyn, Poland}
\email{anna.muranova@matman.uwm.edu.pl}
\author[W. Woess]{Wolfgang Woess}
\address{TU Graz, Institut f\"ur diskrete Mathematik, Steyrergasse 30/III,
8010 Graz, Austria} 

 \email{woess@tugraz.at}

\title[Networks with complex weights]{Networks with complex weights: Green function and power series}

\begin{document}

\begin{abstract}
We introduce a Green function and analogues of other related kernels for finite and infinite networks whose edge 
weights are complex-valued admittances with positive real part. We provide comparison 
results with the same kernels associated with corresponding reversible Markov chains, i.e., 
where the edge weights are positive. Under suitable conditions, these lead to comparison of series of 
matrix powers which express those kernels.
We show that the notions of transience and recurrence extend by analytic continuation to 
the complex-weighted case even when the network is infinite. 
Thus, a variety of methods known for Markov chains extend to that setting.
\end{abstract}

\thanks{Partially funded by Austrian Science Fund FWF-P31889-N35. This work started when the 
first author held a position at TU Graz.}

 %Discrete potential theory
                  %94C15 %Applications of graph theory to circuits and networks
                  %60J05}%Markov processes with discrete parameter

%impedance, $\lambda$-harmonic function, Poisson kernel, distribution, boundary integral} 

\maketitle

\section{Introduction}

A finite or countably infinite connected graph whose edges carry positive real weights can be considered as an  electrical network with resistors, and this is closely related with the intensively studied field of random walks on graphs. See \cite{DS}, \cite{Zem}, \cite{Woess00} and \cite{LP}. In \cite{AlonsoRuiz}, \cite{Baez}, \cite{ChenTeplyaev} and \cite{MuranovaThesis}, the wider class of networks with resistors, coils, and capacitors are considered as complex-weighted graphs.  In the present note,  we use the corresponding model from \cite{Muranova1} and \cite{MuranovaThesis}, i.e  we assume that $(V,E)$ is a connected, locally finite graph without loops, where each (non-oriented) edge  $[x,y]$ is equipped with an \emph{admittance}
\begin{equation}\label{eq:admittance}
\rho_s(x,y)=\rho_s(y,x)=\dfrac{s}{L_{xy}s^2+R_{xy}s+D_{xy}}, \quad x,y\in V,
\end{equation}
where $L_{xy},R_{xy},D_{xy}\ge 0$ with $L_{xy}+R_{xy}+D_{xy}> 0$, and $s \in\C$. 
Here, $L_{xy}$ is the inductance, $R_{xy}$ the resistance and $D_{xy}$ the capacitance
of the edge, and $\rho_s(x,y)$ is the inverse of the impedance. 
%Usually, we shall omit the subscript $s$ in our notation. \texttt{\begin{large}                                                                   \end{large}}
From a viewpoint of Physics, $s$  is a complex frequency, and the admittance of an edge is
the complex-valued analogue of a conductance. Indeed, when $s>0$ is real, $\rho_s(x,y)$ 
can be interpreted as a conductance of the underlying edge.

In the present paper we consider exclusively the case $s \in \mathbb{H}_r\,$, 
the right half plane consisting of all complex nunbers with $\Re s>0$. 
We set $\rho_s(x,y)\equiv 0$ if $[x,y]$ is not an edge, so that $\rho_s$ is a function on $V^2$. 
We call the couple $(V,\rho_s)$ a \emph{complex (electrical) network}.
%For the sake of brevity we refer to a network just by set of vertices $V$.

We introduce the \emph{admittance operator} $P_s\,$, which acts on functions $f:V\rightarrow \C$ 
as follows:
\begin{equation}\label{eq:Ps}
\begin{gathered}
P_s f(x)=\sum_y p_s(x,y)f(y)\,, 
\quad \text{where}\\ %\quad 
p_s(x,y) = \dfrac{\rho_s(x,y)}{\rho_s(x)} \quad\text{with}\quad
\rho_s(x)=\sum_y \rho_s(x,y)\,.
\end{gathered}
\end{equation}
% Here, the \emph{basic assumption} is that $\rho(x)$ is finite for 
% each $x \in  V$. This is of course satisfied when the graph is locally finite (i.e., each vertex
% has finite degree). When $x$ has infinitely many neighbours (countably many, since $V$ must 
% be countable) then the resulting series has to converge absolutely.
The admittance \eqref{eq:admittance} is a positive-real function, that is, $\Re \rho_s(x,y)>0$ when $\Re s>0\,$; 
see \cite{Brune}, \cite{MuranovaThesis}. Therefore $P_s f(x)$ is well-defined at any vertex of our graph. When  
$s\in\R_+\,$, we see that $P_s$ is a stochastic transition matrix which governs a nearest neighbour random 
walk. This is also true when all the vectors $(L_{xy},R_{xy},D_{xy})$ are collinear (proportional). In particular, 
if they are same on each edge, then $P_s$ is the transition matrix of the simple random walk on the graph, 
independently of $s$.

The main questions addressed in this note are threefold:
\begin{itemize}
 \item How can the concept of transience (resp. recurrence) be formulated$\,$?
 \item In the transient case, how can one construct (the analgoue of) the Green function$\,$?
 \item To which extent can the latter be computed in terms of power series$\,$? 
 \end{itemize}
We analyse the analogues of the different Laplace type equations 
associated with $P_s$ when $s$ is complex, as compared to the well-understood case when
it is real.

We first prove, resp. recall some basic estimates of admittances in Section~\ref{section:estimates}. 
In Section~\ref{section:finiteGreen}, we introduce the Green function for finite networks with 
boundary, a non-empty subset of the vertex set where the network is grounded.
We relate the Green function, resp., the analogues of escape probabilities, with the effective 
impedance defined in \cite{Muranova1}, \cite{MuranovaThesis}. It is convenient to work with the 
inverse of effective impedance, that is, the \emph{ effective admittance,} which corresponds 
to the total amount of current in the electrical network. 
In this context, we provide first comparisons of associated power series with analogous ones for
reversible Markov chains.

Our main effort concerns infinite networks, and in Section~\ref{section:infiniteGreen}, 
we study the effective admittance both in presence of a %finite 
boundary $\partial V \subsetneq V$ 
as well as the effective admittance between a source vertex and infinity. The latter leads to the 
notion of transience, resp. recurrence, and our main result is that this does not depend 
on the parameter $s\in \mathbb{H}_r\,$, and that in the transient
case, one always can construct a Green kernel in extension of the well-understood case when $s > 0$. 
In the final Section~\ref{sec:trees}, show how that Green kernel can be used when the network
is a tree. We construct the Martin kernel and provide a Poisson type integral representation
of all harmonic functions over the boundary at infinity of the tree.
In the specific case of a free group, we have a closer look at the applicability of our
comparison results between the complex network and the ones associated wiith positive real weights.

\begin{comment}
In Section~\ref{Section: other power series} we consider further power series for infinite networks and relate them with the effective admittance of an infinite network, as defined in \cite{Muranova3}, \cite{MuranovaThesis}. Then, in Section~\ref{Section: power series on trees} we investigate these power series on infinite trees. More precisely, in Section \ref{Section: lambda} we study \emph{$\lambda$-harmonic} functions, i.e., functions $h:V\rightarrow\C$ satisfying $P_s h=\lambda\cdot h$,  for sufficiently large enough $\lambda \in \C$. We show that every $\lambda$-harmonic function has a boundary integral representation over the geometric boundary at infinity of the tree: we introduce a $\lambda$-Martin kernel on trees with admittances. The integral is taken with respect to a distribution on the boundary of the tree. This generalises the approach of \cite{PicardelloWoess2}, \cite{PicardelloWoess} concerning random walks on trees. 
\end{comment}

\section{Inequalities for admittance operators}\label{section:estimates}

\noindent
\textbf{Notational convention.} In the sequel, we shall compare the complex-weighted
admittance operators $P_s$ with non-negative, stochastic transition operators. 
In order to better visualize these different types, we shall use slightly different fonts: 
$\PQ$ and $\pq(x,y)$ will refer to stochastic transition operators -- even though 
$\PQ_s = P_s$ when $s > 0$.

\begin{lemma}\label{lem:estimates}
The admittance \eqref{eq:admittance} of any edge is a positive-real function of $s$. 
The following estimates hold.  
\begin{align}
|\rho_s(x,y)| &\le \dfrac{|s|}{\Re s}\,\rho_{|s|}(x,y),\label{eq:estimateAdm1}
\\
|\rho_s(x,y)| &\le \dfrac{|s|}{\Re s}\,\Re \rho_s(x,y)\,\label{eq:estimateAdm2}
%\\
%|p_s(x,y)| \le  \Bigl(\dfrac{|s|}{\Re s}\Bigr)^2 \pq_{|s|}(x,y)\label{eq:p|s|}\,.
%{|\rho(x,y)|} &\ge  \dfrac{1}{R_{xy}+L_{xy}+D_{xy}}\min\left(|s|,\frac{1}{|s|}\right).
%\label{estimateAdm3}
\end{align}
\end{lemma}

\begin{proof} We first reconsider the property of being positive-real. 
Note that for any complex number $s\in \C$, $\Re s>0$ if and only if $\Re 1/s > 0$. We have
\begin{equation*}
\Re \frac{1}{\rho_s(x,y)}=L_{xy}\Re s+R_{xy}+D_{xy}\Re\frac{1}{s}>0, \mbox{ whenever }\Re s>0.
\end{equation*}
Next, note that
\begin{equation}\label{eq:trivial}
|\rho_s(x,y)| \ge \rho_{|s|}(x,y). 
\end{equation}
Also note that
$$
\Re\frac{1}{\rho_s(x,y)} = L_{xy} \Re s + R_{xy} + D_{xy} \frac{\Re s}{|s|^2} \ge 
\frac{\Re s}{|s|} \, \frac{1}{\rho_{|s|}(x,y)}.
$$
Therefore
$$
\frac{1}{|\rho_s(x,y)|} \ge \Re \frac{1}{\rho_s(x,y)} 
\ge \frac{\Re s}{|s|} \, \frac{1}{\rho_{|s|}(x,y)}\,,
$$
which proves \eqref{eq:estimateAdm1}
Regarding  \eqref{eq:estimateAdm2}, we use that for $z \in \C$, one has $\Re z = |z|^2 \Re 1/z$.
Thus
\begin{align*}
\Re\rho_s(x,y)&= |\rho_s(x,y)|^2 \,\Re \frac{1}{\rho_s(x,y)} 
\ge |\rho_s(x,y)|^2 \, \frac{\Re s}{|s|}\, \frac{1}{\rho_{|s|}(x,y)}\\ 
&\ge 
|\rho_s(x,y)|\, \frac{\Re s}{|s|}\,,
\end{align*}
and \eqref{eq:estimateAdm2} holds. 
\end{proof}

In addition to the operators (matrices) $P_s$ (resp. $\PQ_s$ when $s > 0$)
we also introduce the 
transition operators $\wt{\!\PQ}_s$ and $\wc{\!\PQ}_s$ with matrix entries
\begin{equation}\label{eq:tildeP}
\wt \pq_{\! s}(x,y) = \dfrac{\Re \rho_s(x,y)}{\Re \rho_s(x)} \AND 
\wc \pq_{\! s}(x,y) = \dfrac{|\rho_s(x,y)|}{|\rho|_s(x)} \,, 
\end{equation}
where $|\rho|_s(x) = \sum_y |\rho_s(x,y)|$. 
From Lemma \ref{lem:estimates}, we get the following comparison.

\begin{proposition}\label{pro:le} For any $s \in \mathbb{H}_r$ and all $x,y \in V$,
$$
\begin{gathered}
|p_s(x,y)| \le  \dfrac{|s|}{\Re s}\, \wt \pq_s(x,y)\,,\quad 
|p_s(x,y)| \le  \dfrac{|s|}{\Re s}\, \wc \pq_s(x,y)\,,
\AND\\
|p_s(x,y)| \le \Bigl(\dfrac{|s|}{\Re s}\Bigr)^2\, \pq_{|s|}(x,y).
\end{gathered}
$$
\end{proposition}

\begin{proof}
Using \eqref{eq:estimateAdm2}, we  obtain
$$
|\rho_s(x)| \ge \Re \rho_s(x) = \sum_y \Re \rho_s(x,y) \ge 
\dfrac{\Re s}{|s|}\sum_y |\rho_s(x,y)| \,,
$$
%\ge  \dfrac{\Re s}{|s|} \rho_{|s|}(x)\,.
and the first two of the proposed inequalities follow. 
Combining the above with \eqref{eq:estimateAdm1} and \eqref{eq:trivial} yields the third one.
\end{proof}

Recall that when $s > 0$ is real, $\PQ_s$ the transition operator of a random walk. This also
holds when all three-dimensional vectors 
$(L_{xy}, R_{xy},D_{xy})$ are collinear, in which case $P_s$ is independent of the
value of $s$. We also have the following comparison.

\begin{lemma}\label{lem:real} If $0 < s < t$ (both real) then for all $x,y \in V$
$$
\begin{aligned}
(s/t)\,\rho_t(x,y) &\le \rho_s(x,y) \le (t/s) \, \rho_t(x,y)\,,\quad \text{whence}\\ 
(s/t)^2\,\pq_{t}(x,y) &\le \pq_{s}(x,y) \le (t/s)^2 \, \pq_{t}(x,y)\,.
\end{aligned}
$$
\end{lemma}
 
\begin{proof} This is elementary: for real $L, R, D \ge 0$ with $L+R+D> 0$, consider
the function 
$$
g(L,R,D) = \frac{Ls + R + D/s}{Lt + R + D/t}\,.
$$
For maximising, resp. minimising $g$, it suffices to consider $L+R+D = 1$, and 
one finds that in the simplex $\{ (L,D) : L+D \le 1, L \ge 0, D \ge 0\}$, the extrema of
$g(L, 1-L-D, D)$ lie in the corners, whence the maximum is $t/s$ and the minimum is $s/t$.
\end{proof}
 
\begin{corollary}\label{cor:s-t} For $s \in \mathbb{H}_r$ and $t > 0$, we have for all $x,y \in V$
$$
|p_s(x,y)| \le \frac{1}{(\Re s)^2} \,\max \Bigl\{ \frac{|s|^4}{t^2}\,,\, t^2 \Bigr\}\,
\pq_t(x,y)\,.
$$
\end{corollary}

(Note the particular case $t=1$.) 
This means that we can investigate some properties of our complex-weighted network via comparison with the corresponding random walks with transition probabilities $p_t(x,y)$, where $t > 0$, or 
$\wt p_s(x,y)$, respectively. 

Notation: in the sequel, we shall write 
$$
\Ppi_+ = \{ \PQ_t\,,\; \wt{\!\PQ}_s\,,\; \wc{\!\PQ}_s : t > 0\,,\; s \in \mathbb{H}_r \}
$$
for the collection of the stochastic matrices that come up in our context, and
$$
\Ppi = \Ppi_+ \cup \{ P_s  : s \in \mathbb{H}_r \}.
$$

\section{The Green function on finite networks with boundaries}\label{section:finiteGreen}

Let $(V,\rho)$ be a finite network. 
We fix a non-empty proper subset $\partial V$ %= \{ a \} \cup B$ 
of $V$. %with  $a\notin B$. 
We consider $V^{\circ}=V\setminus \partial V$ % (\{ a \} \cup B)$ 
as the \emph{interior} of our graph. 

If $P = \bigl( p(x,y) \bigr)_{x,y \in V}$ is any real or complex matrix indexed by $V$, then we let 
$$
P_{V^{\circ}} =  \bigl( p(x,y) \bigr)_{x,y \in V^{\circ}} \quad \text{and}\quad
P_{V^{\circ}\!,\,\partial V} =  \bigl( p(x,y) \bigr)_{x\in V^{\circ},\, y \in \partial V}\,.
$$
We write $P_{V^{\circ}}^n =  \bigl( p_{V^{\circ}}^{(n)}(x,y) \bigr)_{x,y \in V^{\circ}}\,$, 
so that in particular, $P_{V^{\circ}}^0 = I_{V^{\circ}}$ is the identity matrix over $V^{\circ}$. 

\begin{definition}\label{def:Green} Whenever the matrix $I_{V^{\circ}} - P_{V^{\circ}}$ is invertible, let
$$
G_{V^{\circ}}^{P} = \bigl( I_{V^{\circ}} - P_{V^{\circ}} \bigr)^{-1}\,. 
$$
Its matrix elements $G_{V^{\circ}}^{P}(x,y)$ are called the \emph{Green function} or 
\emph{Green kernel} of $P$ with respect to the chosen interior $V^{o}$.
\end{definition}

Since each of the stochastic matrices $\PQ \in \Ppi_+$
is irreducible, it is a quite elementary fact that $G_{V^{\circ}}^\PQ$ exists; see e.g.  
\cite[Lemma 2.4]{HiWo}. For the Markov chain with transition matrix $\PQ$ starting at 
vertex $x$, we have that $G_{V^{\circ}}^{\PQ}(x,y)$ is the expected
number of visits in $y$ before leaving the interior $V^{\circ}$. 
%The same holds for $\widetilde \PQ$. 
Furthermore, it follows from \cite{Muranova3} and \cite{Draper} that also
for complex weights with positive real part,
$I_{V^{\circ}} - P_{s|V^{\circ}}$ is invertible for every $s \in \mathbb{H}_r\,$. 
See in particular the proof of \cite[Theorem 2]{Muranova3}. 

\begin{definition}\label{def:Lap} If $P \in \Ppi$, then the associated 
\emph{normalized weighted Laplacian} 
$\Delta_P$ is the operator acting on functions 
$f:V\rightarrow \C$ by
\begin{equation*}
\Delta_{P} f(x)=\sum_{y:y\sim x}\bigl(f(y)-f(x)\bigr)p(x,y).%{\widetilde \rho} (x,y)
%=\sum_{y:y\sim x}(\nabla_{xy} f)\widetilde\rho(x,y),
\end{equation*}
%where
%\begin{equation*}
%\nabla_{xy} f=f(y)-f(x)
%\end{equation*}
%is the \emph{difference operator}.
\end{definition}
 A function $v:V\rightarrow \C$ is 
called \emph{harmonic on $V^{\circ}$ with respect to $\Delta_{P}$} if 
\begin{equation*}
\Delta_{P} \, v(x)=0
\end{equation*}
for any $x\in V^{\circ}$. Now choose $a \in V^{\circ}$ and consider the \emph{augmented boundary} 
$\partial^a V = \partial V \cup \{a\}$ as well as the \emph{reduced interior} $V^a = V^{\circ} \setminus \{ a \}$.
Harmonic functions come up in the following \emph{Dirichlet problem.}
\begin{equation}\label{eq:dirpr}
 \begin{cases}
\Delta_{P} \, v(x)=0 \;\mbox { on }\; V^a,
   \\
   v(a)=1, 
 \\ 
v\raisebox{-.5ex}{$|$}_{\partial V}\equiv 0.
\end{cases}
\end{equation}
Our interest is in $P=P_s$ and the associated Dirichlet problem with complex weights.  
By \cite{Muranova3} and \cite{Draper} this problem has a unique solution $v = v^a$ whenever $\Re s>0$. 
Indeed, the function $v\raisebox{-.5ex}{$|$}_{\partial^a V}$ provides the (augmented) boundary data, 
and the solution can be given in two ways: 
%For $P = P_s\,$,
\begin{align}
\hspace*{-3cm}\text{For }\; P = P_s\,,\qquad 
v\raisebox{-.5ex}{$|$}_{V^a} &= G_{V^a}^{P} \,\, 
P_{V^a,\,\partial^a V} \,\,\, v\raisebox{-.5ex}{$|$}_{\partial^a V} \label{eq:Dirsol1} \\
\hspace*{-3cm}\text{and }\phantom{\; P = P_s\,,\qquad} 
v\raisebox{-.5ex}{$|$}_{V^{\circ}} &= G_{V^{\circ}}^{P}(\cdot,a)/G_{V^{\circ}}^{P}(a,a)
\label{eq:Dirsol2}
\end{align}
where (as usual) functions are to be seen as column vectors. Indeed, one easily checks that 
both formulas provide a solution of \eqref{eq:dirpr}, and by uniquenss, they coincide.

The Dirichlet problem has a physical interpretation. In the electrical network model, 
the vertex $a$ is the source, where the potential is kept at $1$, and $\partial V$ is the set 
of grounded nodes. Then $v(x)$ is the complex voltage at the vertex $x$ (for the complex frequency $s$).
This leads to the following definition.

\begin{definition}\label{def:admit}\cite{Muranova1},\cite{Muranova3}\label{ZeffC}
For the finite network $(V,\rho_s)$ with $s \in \mathbb{H}_r\,$,
the \emph{admittance\footnote{In \cite{Muranova3}, the symbol $\mathcal P$ is used for 
the admittance of the network.} %$\P(a\rightarrow \partial V) = \P_s(a\rightarrow \partial V)$
between the source vertex $a$ and the grounded set $\partial V$} is defined by 
\begin{align*}
\P_s(a\rightarrow \partial V) &= \sum_{x: x \sim a} \bigl(1-v(x)\bigr)\rho_s(a,x)
=\sum_{b\in \partial V}\sum_{y:y\sim b}v(y)\rho_s(y, b)\\ 
&= \frac12 \sum_{x,y \in V} |v(x)-v(y)|^2 \, \rho_s(x,y)\,,
\end{align*}
where $v(x)= v^a(x)$ is the solution of the Dirichlet problem \eqref{eq:dirpr} with 
respect to $\Delta_{P_s}\,$.
\end{definition}

When $s > 0$, this is of course classical, and $\P_s(a\rightarrow \partial V)$
is the inverse of the \emph{total resistance} between $a$ and $\partial V$, while the 
resistance of a single edge is $1/\rho_s(x,y)$.
The following is immediate from formula \eqref{eq:Dirsol2}.

\begin{lemma}\label{lem:solgreen}  \hspace*{2.2cm}
$\displaystyle 
\P_s(a\rightarrow \partial V) = \frac{\rho_s(a)}{G_{V^{\circ}}^{P_s}(a,a)}\,.
$
\end{lemma}

Let us have another look at  \eqref{eq:Dirsol1} and \eqref{eq:Dirsol2}. If
we replace the complex matrix $P_s$ by a stochastic matrix $\PQ \in \Ppi_+$ 
then we have the Markov chain $(X_n)_{n \ge 0}$ with transition matrix $\PQ$.
Given $\partial V$ and $V^{\circ}$, we can consider the stopping time of the first visit in $a$ 
before leaving $V^{\circ}\,$:
$$
\tb^a = \inf \{n \ge 0 : X_n = a\,,\; X_k \in V^{\circ} \; \text{for all}\; k \le n \} \le \infty.
$$
For $x \in V^{\circ}$, set 
$$
F_{V^{\circ}}^{\PQ}(x,a) = \mathbb{P}[\tb^a < \infty \mid X_0=x] = 
\sum_{n=0}^{\infty} \underbrace{\mathbb{P}[\tb^a = n \mid X_0=x]}_{\displaystyle
=:\fq_{V^{\circ}}^{(n)}(x,a)}.
$$
Note that $f_{V^{\circ}}^{(n)}(a,a) = \delta_0(n)$ and that $f_{V^{\circ}}^{(0)}(x,a)= 0$ for $x \in V^a\,$.
It is well-known and easy to prove that
$$
F_{V^{\circ}}^{\PQ}(x,a) = G_{V^{\circ}}^{\PQ}(\cdot,a)/G_{V^{\circ}}^{\PQ}(a,a)\,,
$$
the solution of our Dirichlet problem when $\PQ \in \Ppi_+\,$.
Furthermore, 
\begin{equation}\label{eq:fn}
\fq_{V^{\circ}}^{(n)}(x,a) = \sum_{w \in V^a}\pq_{V^a}^{(n-1)}(x,w)\,\pq(w,a) \quad 
\text{for }\;x \in V^a\,, \; n \ge 1.
%
%= \sum_{w \in V^a} G_{V^a}^{\PQ}(x,w)p(w,y) = 
%\sum_{n=0}^{\infty} p_{V^a}^{(n)}(x,w)p(w,y)\,.
\end{equation}
%We have the analogous quantities for $\widetilde \PQ$, denoted    $\widetilde p_{V^{\circ}}^{(n)}\,$, $\widetilde \tb\,$, and so on. 
The estimates of \S \ref{section:estimates} suggest that we can compare the solution
of \eqref{eq:dirpr} and related items concerning the complex network $(V,\rho)$ with 
the analogous ones for $\PQ \in \Ppi_+\,$. For any $P \in \Ppi$ (i.e., including complex weights), 
we introduce the power series
\begin{equation}\label{eq:Green}
G_{V^a}^P(x,y|z) = \sum_{n=0}^{\infty}p_{V^a}^{(n)}(x,y)\,z^n\,, \quad x,y \in V^a\,,\; z \in \mathbb{C}.
\end{equation}

\begin{proposition}\label{pro:compare}
\emph{(i)} for $\PQ \in \Ppi_+\,$, let 
$\lambda(\PQ_{V^a}) = \max \{ |\lambda| : \lambda \in \textsf{spec}(\PQ_{V^a})\}$.
Then $\lambda(\PQ_{V^a}) < 1$, it is an eigenvalue of $\PQ_{V^a}\,$, and for 
$|z| < 1/\lambda(\PQ_{V^a})$, each of the power series of \eqref{eq:Green} converges absolutely. %With the analogous definition
%of $\lambda(\widetilde P_{V^{\circ}})$, the corresponding  statement holds for $\widetilde \PQ$. 
\\[5pt]
\emph{(ii)} If for complex $s \in \mathbb{H}_r\,$, the series 
\begin{equation}\label{eq:series}
v^a(x) = \sum_{n=0}^{\infty} \sum_{y \in V^a} p_{s|V^a}^{(n)}(x,y)p_s(y,a)
\end{equation}
converges absolutely for every $x  \in V^a$, then it is the solution of the Dirichlet problem \eqref{eq:dirpr}.
\\[3pt]
This holds whenever for $t > 0$
\begin{equation}\label{eq:zs}
\begin{aligned}
r_{s,t} &= \frac{1}{(\Re s)^2} \,\max \Bigl\{ \frac{|s|^4}{t^2}\,,\, t^2 \Bigr\}
< \dfrac{1}{\lambda(\PQ_{t|V^a})}
\quad \text{or} \\ 
 r_{s} &= \dfrac{|s|}{\Re s} < \dfrac{1}{\lambda(\,\wt{\!\PQ}_{s|V^a})}\,,\quad 
 \text{or} \quad r_s < \dfrac{1}{\lambda(\,\wc{\!\PQ}_{s|V^a})}\,.
\end{aligned}
\end{equation}
In these cases, the series \eqref{eq:series} is dominated in absolute value by 
$$
\begin{gathered}
\sum_{y \in V^a} G_{V^a}^{\PQ_t}(x,y|r_{s,t})\, \pq_t(y,a)\, r_{s,t}\,, 
\quad\text{resp.}\\
\sum_{y \in V^a} G_{V^{\circ}}^{\,\,\wt{\!\!\PQ}_s}(x,y| r_{s})\, 
\widetilde \pq_s(y,a)\, r_{s}\,, \quad \text{resp.}\quad 
\sum_{y \in V^a} G_{V^{\circ}}^{\,\,\wc{\!\!\PQ}_s}(x,y| r_{s})\, 
\wc \pq_s(y,a)\, r_{s}\,.
\end{gathered}
$$
\end{proposition}

\begin{proof} (i) The subgraph induced by ${V^a}$ has one or more connected components 
$C_1\,,\dots, C_k\,$.  
Each of the corresponding sub-matrices $\PQ_{C_i}$ of $\PQ$ is irreducible and 
non-negative,  and these matrices give rise to a block-decomposition of  $\PQ_{V^{\circ}}\,$. 
By the Perron-Frobenius theorem, the spectral radius of $\PQ_{C_i}$ coincides with 
its largest eigenvalue, which is positive real. It is $< 1$, since $\PQ_{C_i}$ is substochastic, 
but not stochastic. The maximum of the Perron-Frobenius eigenvalues of all the matrices 
$\PQ_{C_i}$ is $\lambda(\PQ_{V^a})$, and the Perron-Frobenius theorem also yields absolute 
convergence of $G_{V^a}^{\PQ}(x,y|z)$ for $|z| < 1/\lambda(\PQ_{V^{\circ}})$  and all $x,y \in V^{\circ}$. 

\smallskip

(ii) If $\sum_n p_{s|V^a}^{(n)}(x,y)$ converges absolutely for all $x, y \in V^{\circ}$ then 
the value of the series is $G_{V^{\circ}}^{P_s}(x,y)$, so that $\eqref{eq:series}$ is indeed 
the solution \eqref{eq:Dirsol1} of the Dirichlet problem. The last part of the proposition 
follows from Proposition \ref{pro:le}, resp. Corollary~\ref{cor:s-t}. 
\end{proof}

In statement (ii) above, the most natural choices for $t$ are $t = |s|$ or $t = 1$. 
The advantage of the comparison lies in the possibility to use combinatorial methods
of generating functions and paths for computing the solution of the 
Dirichlet problem \eqref{eq:dirpr}.

\begin{example}\label{ex:ample}
We consider the graph with verex set $V = \{1,2,3,4\}$  as in Figure \ref{fig:example}. We choose 
$s = e^{i\alpha}$ where $|\alpha| < \pi/2$. Along each edge, the label in the figure is its admittance. 

\begin{figure}[H]
\centering
\begin{tikzpicture}[auto,node distance=2.5cm,
                    thick,main node/.style={circle, draw, fill=black!100,
                        inner sep=0pt, minimum width=3pt}]

  \node[main node] (1) [label={[below=3pt]$1$}]{};
  \node[main node] (2) [above of=1,label={[above]$2$}] {};
  \node[main node] (3) [right of=2,label={[above]$3$}] {};
  \node[main node] (4) [below of=3,label={[below=3pt]$4$}] {};

  \path[every node/.style={font=\sffamily\small}]
    (1) edge node [bend right] {$\dfrac{1}{s}$} (2)
    (2) edge node [bend right] {$s$} (3)
    (3) edge node [bend right] {$\dfrac{1}{s}$} (4)
    (4) edge node [bend right] {$s$} (1)
    (1) edge node [bend right] {$s$} (3);

\end{tikzpicture}
\caption{}%caption}
\label{fig:example}
\end{figure}

We have 
$$
P_s = \begin{pmatrix} 0 & \frac{1}{2s^2+1} & \frac{s^2}{2s^2+1} & \frac{s^2}{2s^2+1} \\[3pt]
                      \frac{1}{s^2+1} & 0  & \frac{s^2}{s^2+1}  & 0  \\[3pt]
                      \frac{s^2}{2s^2+1}   & \frac{s^2}{2s^2+1} & 0 & \frac{1}{2s^2+1}\\[3pt]
                      \frac{s^2}{s^2+1}    & 0 & \frac{1}{s^2+1}  & 0  
      \end{pmatrix}
\quad \text{and}\quad \PQ_1 = 
      \begin{pmatrix} 0 & \frac13 & \frac13 & \frac13 \\[3pt]
                      \frac12 & 0 & \frac12 & 0 \\[3pt]
                      \frac13 & \frac13 & 0 & \frac13 \\[3pt]
                      \frac12 & 0 & \frac12 & 0       
      \end{pmatrix}.
$$
Also, $\wt{\!\PQ}_s = \wc{\!\PQ}_s = \PQ_1\,$, since $|s|=1$. 
We consider $a=1$ and $\partial V = \{4\}$, so that $V^a=\{2,3\}$. For our comparison, we choose 
$t = |s| = 1$ and write $\PQ = P_1$. Then $\lambda(\PQ_{V^a}) = 1/\sqrt{6}$.  
Also, $r_{s,1} = 1/\cos^2 \alpha$ and $\widetilde r_{s} = 1/\cos \alpha$. The latter is better (smaller) than the former. We see that the comparison of Proposition 
\ref{pro:compare} works whenever $\cos \alpha > 1/\sqrt{6}$. On the other hand, the spectral radius of $P_{s|V^a}$ satisfies
$$
\lambda(P_{s|V^a})^2 = \frac{1}{|2s^2+1| |s^2+1|}
$$
One gets that $|\lambda(P_{s|V^a})| < 1$ if and only if $\cos \alpha > 1/\sqrt{8}$, and precisely in
this case, the series \eqref{eq:series} converges absolutely, while the comparison with 
$\PQ_1$ (resp. $\wt{\!\PQ}_s$ or $\wc{\!\PQ}_s$) is not useful when  
$1/\sqrt{8} < \cos \alpha \le 1/\sqrt{6}$. 
Finally, if $\cos \alpha \le 1/\sqrt{8}$, the series \eqref{eq:series} diverges and cannot be 
used for solving the Dirichlet problem. We also remark that for $t > 0$ comparison with $\PQ_t$ 
yields no improvement when $t \ne 1$.
\end{example}

\section{Admittance and Green function on infinite networks}\label{section:infiniteGreen} 

%\textbf{B. Green function on infinite networks}
%\\[3pt]
The above can also be done when the network is infinite. Recall that we assume local finiteness
(each node has finitely many neighbours). In this case, the set $\partial V \subsetneq V$ 
of grounded states may be finite or infinite. It may also be empty, in which case we are considering a  
complex-valued flow from $a$ to $\infty$. (Indeed, the boundary should rather be thought of as
$\partial V \cup \{\infty\}$.) 
We can again consider the power series \eqref{eq:Green}.
When is $\partial V$ non-empty, we get that
\begin{equation}\label{eq:fingrn}
\PQ \in \Ppi_+ \quad \Rightarrow \quad G_{V^{\circ}}^{\PQ}(x,y) = G_{V^{\circ}}^{\PQ}(x,y|1) < \infty\,,
\end{equation}
because $\partial V$ is a set of absorbing states for this Markov chain.  
%The same holds for $\widetilde \PQ$. 
In addition, we can also consider the unrestricted Green function 
\begin{equation}\label{eq:Gxyz}
G^{\PQ}(x,y|z) = \sum_{n=0}^{\infty}\pq^{(n)}(x,y)\,z^n\,, \quad x,y \in V\,,\; z \in \mathbb{C}\,,
\end{equation}
and when $z=1$ and the series converges absolutely, we write once more $G^{\PQ}(x,y) = G^{\PQ}(x,y|1)$. 
Finiteness of $G^{\PQ}(x,y)$  means that the associated
Markov chain with transition matrix $\PQ \in \Ppi_+$ is \emph{transient:} 
with probability $1$, each vertex is visited only finitely
often by the random process, and finiteness is independent of $x$ and $y$ by connectedness
of the graph. More generally, consider the \emph{spectral radius}
\begin{equation}\label{eq:specrad}
\lambda(\PQ\,) = \limsup_{n \to \infty}\, \pq^{(n)}(x,y)^{1/n}\,. 
\end{equation}
It is well known that this number is independent of $x$ and $y$. It is indeed the spectral radius (norm) of $\PQ$ acting as a self-adjoint operator on $\ell^2(V, m)$, where the weights are 
\begin{equation}\label{eq:m}
\begin{aligned}
m(x) &= \begin{cases}
        \rho_t(x)  &\text{if }\; \PQ=\PQ_t\,,\; t > 0,\\
        \Re \rho_s(x) &\text{if }\; \PQ= \wt{\!\PQ}_s\,,\\
        |\rho|_s(x) &\text{if }\; \PQ= \wc{\!\PQ}_s\,,\; s \in \mathbb{H}_r\,,  
       \end{cases}\\
\| f\|_m^2 &= \sum_{x \in V}  |f(x)|^2\,m(x).
\end{aligned}
       \end{equation}

Furthermore, the radius of convergence
of the power series $G^{\PQ}(x,y|z)$ is $1/\lambda(\PQ\,)$, and at $z = 1/\lambda(\PQ\,)$ the latter either converges for all $x,y$ or diverges for all $x,y$. In the first of those two cases, $\PQ$ is called \emph{$\lambda(\PQ\,)$-transient,} in the second case \emph{$\lambda(\PQ\,)$-recurrent.} See e.g. \cite{Woess09}. 
%The same holds for $\widetilde \PQ$. 
In the case of a finite network, we have of course $\lambda(\PQ\,) = 1$, and the respective 
Green function diverges at $z=1$. 

\begin{proposition}\label{pro:transrec}
\emph{(a)} The Markov chains induced by $\PQ \in \Ppi_+$ are either all recurrent or all
transient.
\\[5pt]
\emph{(b)} We either have $\lambda(\PQ\,)=1$ for all $\PQ \in \Ppi_+$ or  
$\lambda(\PQ\,)<1$ for all $\PQ \in \Ppi_+\,$.
\end{proposition}

\begin{proof} Let $\ell_0(V)$ be the space of all finitely supported real or complex
functions on $V$. With $m(x)$ as in \eqref{eq:m}, the \emph{Dirichlet sum} associated
with $\PQ \in \Ppi_+$ is
$$
D_{\!\PQ\,}(f) = \frac12 \sum_{x,y \in V} |f(x)-f(y)|^2\, m(x)\pq(x,y).
$$
By Proposition \ref{pro:le} and Corollary~\ref{cor:s-t}, the Dirichlet sums associated
with distinct $\PQ, \QQ \in \Ppi_+$ compare above and below by positive multiplicative constants.
Thus, by \cite[Corollary 2.14]{Woess00} (to cite one among various sources), transience of $\PQ$
implies transience of $\QQ$ and vice versa. This proves (a).
 
Regarding the spectral radius, by \cite[Theorem 10.3]{Woess00} (once more to cite one among various sources) we have $\lambda(\PQ\,) < 1$ if and only if there is $\overline \kappa > 0$ such that
$$
\|f\|_m^2 \le \overline \kappa \, D_{\!\PQ\,}(f) \quad \text{for all }\; f \in \ell_0(V).
$$
Besides the Dirichlet forms, also the weights $m$ with respect to different $\PQ, \QQ \in \Ppi_+$ 
compare up to positive multiplicative constants. This yields (b). 
\end{proof}

Let us now consider the effective admittance of our infinite network, as defined in 
\cite{Muranova3} and \cite{MuranovaThesis}. 
Let $V_n = \{ x \in V : d(x,x_0) \le n \}$ be the ball of radius $n$ around a choosen 
the root vertex $x_0$ with respect to the integer-valued graph metric, and let $E_n$ be 
the set of edges whose endpoints lie in $V_n\,$. Thus, $(V_n\,,E_n)$ is the subgraph 
of $(V,E)$ induced by $V_n$. We write $(V_n\,,\rho_s)$ for the resulting finite sub-network
of $(V,\rho_s)$, where more precisely, the admittance function $\rho_s = \rho_{s,n}$ is the restriction 
of the given one to $E_n\,$.
Given the %finite 
set of grounded states $\partial V$ in the infinite network, as well as 
an input node $a \in V \setminus \partial V$, we take $n$ large enough 
such that $a \in V_{n-1}$ and define 
\begin{equation}\label{eq:bdVn}
\partial V_n = (V_{n-1} \cap \partial V) \cup S_n\,, 
\quad \text{where}\quad S_n = V_n \setminus V_{n-1}\,. 
\end{equation}
Now let $v_n(x) = v_{s,n}^a(x)$ be the unique solution of the Dirichlet problem 
\ref{eq:dirpr} on $(V_n,\rho_s)$ with respect to $(a,\partial V_n)$, as given by 
\eqref{eq:Dirsol1} and \eqref{eq:Dirsol2}.
Its dependence on $s$ is important here. The associated 
effective admittance is
$$
\P_s(a \rightarrow \partial V_n) =
\sum_{x: x \sim a} \bigl(1 - v_n(x)\bigr)\rho_s(a,x).
$$
\begin{proposition}\label{pro:converge}\cite[Thm. 22]{Muranova3}  \label{pro:limit} As a function of $s \in \mathbb{H}_r\,$, the sequence $\bigl(\P_s(a \rightarrow \partial V_n)\bigr)_n$ converges locally uniformly to a 
holomorphic function:
$$
\P_s(a \rightarrow \partial V \cup \{\infty\}) 
= \lim_{n \to \infty} \P_s(a \rightarrow \partial V_n).
$$
The limit is the \emph{effective admittance} of the infinite network.
\end{proposition}

We are led to the following.

\begin{definition}\label{def:transience}
Given the parameters $(L_{xy},R_{xy},D_{xy})$ and the admittances $\rho^{(s)}(x,y)$ on all edges 
of $(V,E)$, where $s \in \mathbb{H}_r\,$, the infinite network $(V, \rho^{(s)})$ is called \emph{transient,} if $\P_s(a \rightarrow \infty) \ne 0$ for some source vertex $a \in V$. Otherwise, it is called 
\emph{recurrent.}
\end{definition}

The definition is motivated by the case $s > 0$, in which case we know that $\PQ_s$ is 
the transition matrix of a reversible Markov chain, or equivalently, a resistive network, 
were the edge resistances are $L_{xy}s+R_{xy}+D_{xy}/s\,$. This Markov chain is transient
(i.e., it tends to $\infty$ almost surely)
if and only if the effective conductance from any vertex $a$ to infinity is positive 
(equivalently, the effective resistance is finite). Based on the previous results of 
\cite{Muranova3}, the following is now quite easy to prove, but striking.

\begin{theorem}\label{thm:transience}
\emph{(a)} Transience (resp., recurrence) is independent of the source vertex $a$ as well as of 
the parameter $s \in \mathbb{H}_r$.
\\[5pt] 
\emph{(b)} If the (finite) set $\partial V$ of grounded nodes is non-empty, then we have 
%comment for editor: ``we have'' is only there to create a new line so that the next formula
%is not broken up in the middle.
$\P_s(a \rightarrow \partial V \cup \{\infty\}) \ne 0$ for all $a \in V$ and $s \in \mathbb{H}_r\,$.
\end{theorem}

\begin{proof} By Proposition \ref{pro:limit}, 
$\P_s(a \rightarrow \partial V \cup \{\infty\})$  is the locally uniform limit of a sequence of real-positive holomorphic functions of the 
variable $s \in\mathbb{H}_r$. Hence it is holomorphic on the right half plane. 
By Hurwitz' Theorem (see e.g. \cite[p. 178]{Ahl}), it is either nowhere zero or 
constant equal to zero on $\mathbb{H}_r$.

\smallskip 

(a) Suppose that $\partial V = \emptyset$.
We know already from Proposition \ref{pro:transrec} that for real $s > 0$, 
transience of the reversible Markov chain with transition matrix $\PQ=\PQ_s$ in independent of $s$.
In this case it is well known that the effective admittance (or rather conductance in this situation)
is $\rho_s(a)/G^{\PQ}(a,a)$. 
Transience then means that $G^{\PQ}(a,a) < \infty\,$. It is also well known that in this case, 
$G^{\PQ}(x,y) < \infty$ for all $x,y \in V$. See e.g. \cite{Woess09} or \cite{LP}. 

Thus, when $G^{\,\PQ_s}(a,a) < \infty$ for some $s > 0$ and $a \in V$, then one also has 
$\P_s(a \rightarrow \infty) \ne 0$ for all $s\in \mathbb{H}_r$ and all $a \in V$. 

\smallskip

The proof of (b) is analogous: when $s > 0$, then $G_{V^{\circ}}^{\PQ_s}(a,a) < \infty$ 
for all $a \in V^{\circ}$,
as observed in \eqref{eq:fingrn}. Again, in this case, the effective admittance
is $\rho(a)/G_{V^{\circ}}^{\PQ_s}(a,a)$, and the extension to complex $s \in \mathbb{H}_r$ 
works as in (a).
\end{proof}

In the transient case (with $\partial V = \emptyset$), if $\PQ=\PQ_s$ for $s > 0$, 
we have by monotone convergence 
$$
G^{\PQ}(a,a) = \frac{\rho_s(a)}{\P_s(a \rightarrow \infty)} = \lim_{n \to \infty} G^{\PQ}_{V_n^{\circ}}(a,a)
$$
and 
$$
F^{\PQ}(x,a) = \lim_{n \to \infty} F^{\PQ}_{V_n^{\circ}}(x,a)\,,
$$
where $F^{\PQ}_{V_n^{\circ}}(x,a) = v_{s,n}(x)$, the solution of the corresponding Dirichlet problem with
source node $a$ and grounded set $\partial V_n\,$, see above. The analogous statement is true
for $s > 0$, when $\partial V$ is non-empty.

In the general case of complex  $s \in \mathbb{H}_r\,$, it is natural to 
\emph{define} the on-diagonal Green kernel by
\begin{equation}\label{eq:diag}
\begin{aligned}
G^{P_s}(a,a) &= \frac{\rho_s(a)}{\P_s(a \rightarrow \infty)} 
\qquad \text{in the transient case, resp.}\\[5pt]
G_{V^{\circ}}^{P_s}(a,a) &= \frac{\rho_s(a)}{\P_s(a \rightarrow \partial V \cup \{\infty\})}
\qquad \text{when }\; \partial V \ne \emptyset.
\end{aligned}
\end{equation}

We shall unify notation, writing $G_{V^{\circ}}^{P_s}(a,a)$ in both cases of \eqref{eq:diag},
so that the index $V^{\circ}$ can be omitted when $\partial V = \emptyset$. This also applies to
the following consideration of the off-diagonal elements. 

\begin{theorem}\label{thm:converge}
For complex $s \in \mathbb{H}_r\,$,
$$
F^{P_s}_{V^{\circ}}(x,a) = \lim_{n \to \infty} 
 F^{P_s}_{V_n^{\circ}}(x,a)  \quad \Bigl( = \lim_{n \to \infty} v_{s,n}^a(x)\ \Bigr)
$$
exists for all $a, x \in V$ and defines a holomorphic function of $s$.
\end{theorem}

\begin{proof}
(Note that when $x = a$, the sequence is constant $=1$.)
We set $\P_n = \P_s(a \rightarrow \partial V_n)$.
%Combining 
Then \cite[Cor. 1]{Muranova3} %with \eqref{eq:estimateAdm1}  
shows that for any $n$, 
$$
|\P_n| \le \frac{|s|^2(1+|s|^2)}{(\Re s)^3} \Re \rho_{1}(a).  
$$
We now use the last identity of Definition \ref{def:admit}, also proved in \cite{Muranova3}.
It yields
$$
\frac12 \sum_{x,y \in V_n} |v_n(x)-v_n(y)|^2 \, \Re \rho_s(x,y) \le |\P_n|.
$$
Note that each edge appears twice in that sum. For each $x \in V\,$, we choose a shortest
path $[a=x_0\,,x_1\,, \dots, x_k=x]$ from $a$ to $x$ in our graph. If $n$ is sufficently large 
then it is contained in $(V_n\,,E_n)$. Recall that $v_n(a)=1$. 
We now use the Cauchy-Schwarz inequality:
$$
\begin{aligned}
&|v_n(x)-1|^2 \le \\ &\quad \le \left( \sum_{j=1}^k |v_n(x_j) - v_n(x_{j-1})| 
\sqrt{\Re \rho_s(x_{j-1},x_j)} \cdot \frac{1}{\sqrt{\Re \rho_s(x_{j-1},x_j)}}\right)^2\\
&\quad\le \left(\sum_{j=1}^k |v_n(x_j) - v_n(x_{j-1})|^2 \,\Re \rho_s(x_{j-1},x_j)\right)
    \cdot \left(\sum_{j=1}^k \frac{1}{\Re \rho_s(x_{j-1},x_j)}\right)\\
&\quad\le |\P_n| \cdot \left(\sum_{j=1}^k \frac{1}{\Re \rho_s(x_{j-1},x_j)}\right).
\end{aligned}
$$
Combinig \eqref{eq:estimateAdm2}, \eqref{eq:trivial} and Lemma \ref{lem:real}, 
we get that for any edge $[x,y]$,
$$
\Re \rho_s(x,y) \ge \frac{\Re s}{\max\{1,|s|^2\}}\, \rho_1(x,y).
$$
Therefore 
$$
\sum_{j=1}^k \frac{1}{\Re \rho_s(x_{k-1},x_k)} \le \frac{1 +|s|^2}{\Re s}\, C(x)\,,
\quad \text{where}\quad
C(x) = \sum_{j=1}^k \frac{1}{\rho_1(x_{k-1},x_k)}\,,
$$
which of course depends on the chosen shortest path from $a$ to $x$.
We conclude that 
$$
|v_n(x)| \le 1 + \sqrt{C(x) \,\Re \rho_1(a)}\, \frac{|s|(1+|s|^2)}{(\Re s)^{2}}.
$$
We can now proceed as in the proof of \cite[Thm. 6a]{Muranova3}.
For any fixed $x$ and $a$ in $V^{\circ}$, the sequence of holomorphic (rational)
functions 
$s \mapsto v_{n,s}(x) =  F^{P_s}_{V_n^{\circ}}(x,a)$
is bounded in any domain $\{ s \in \C : \Re s > \varepsilon, |s| < c \}$ 
where $0 < \varepsilon < c$. 
By Montel's theorem \cite[p. 153]{Conway}, this sequence of functions is precompact
in  $\mathbb{H}_r$ with respect to uniform convergence on compact sets. The limit of any convergent subsequence must me holomorphic in $\mathbb{H}_r\,$.
Now, if $s > 0$ is real, then 
$$
\lim_{n \to \infty} F^{\PQ_s}_{V_n^{\circ}}(x,a) = F^{\PQ_s}_{V^{\circ}}(x,a)
$$
by monotone convergence. But a holomorphic function on $\mathbb{H}_r$ is determined
by its values on the positive real half-axis. Therefore we have convergence on 
all of $\mathbb{H}_r\,$.
\end{proof}

\begin{corollary}\label{cor:rec}
 In the recurrent case (with $\partial V = \emptyset$),
 $$
 F^{P_s}(x,a): = \lim_{n \to \infty} F^{P_s}_{V_n^{\circ}}(x,a) = 1
 $$
 for all $x, a \in V$ and all $s \in \mathbb{H}_r\,$.
\end{corollary}

This holds once more by Montel's theorem, since $F^{P_s}(x,a)=1$ for all $s > 0$
(stochastic case).

We now can  define the Green kernel of the transient infinite network by
\begin{equation}\label{eq:allgreen}
G^{P_s}(x,a) = F^{P_s}(x,a)\,G^{P_s}(a,a)\,,\quad s \in \mathbb{H}_r\,,
\end{equation}
with $G^{P_s}(a,a)$ given by \eqref{eq:diag}. (Recall that $F^{P_s}(x,a)=1$.)
Then, in matrix notation,  
$$
(I_V - P_s) G^{P_s} = I_V\,,
$$
where $I_V$ is the identity matrix over $V$. 
In precisely the same way, we also get the Green kernel $G^{P_s}_{V^{\circ}}(x,a)$
when $\partial V \ne \emptyset$, and it satisfies 
$$
(I_{V^{\circ}} - P_{s|V^{\circ}}) G^{P_s}_{V^{\circ}} = I_{V^{\circ}}\,.
$$
\begin{question}\label{que:converge}
Is it true that in the transient case, the analogue of the last identity of
Definition \ref{def:admit} holds for the infinite network\,? That is,
setting $v(x) = F^{P_s}(x,a) = \lim_n v_n(x)$, then is it true that  
$$
\P_s(a\rightarrow \infty)  = \frac12 \sum_{x,y \in V} |v(x)-v(y)|^2 \, \rho_s(x,y)\,,\quad 
s \in \mathbb{H}_r\;?
$$
For $s > 0$, this is well-known to hold, see e.g. \cite[Exercise 2.13]{Woess00}.
\end{question}

Let us take up the notation of \eqref{eq:fn}, for arbitrary $P \in \Ppi\,$:
\begin{equation}\label{eq:fk}
f_{V^{\circ}}^{(k)}(x,a) = \sum_{w \in V^a}p_{V^a}^{(k-1)}(x,w)p(w,a) \quad 
\text{for }\;x \in V^a\,, \; k \ge 1.
\end{equation}
Once more, $V^{\circ} = V$ when $\partial V = \emptyset$. By local finiteness, the sum is finite.
We use the analogous notation with respect to $V_n^{\circ}$ and $V_n^a$, where  $V_n\,$, $n \in \N$,
are the vertex sets of our increasing family of finite subnetworks. We also consider the generating
power series 
\begin{equation}\label{eq:FFn}
\begin{aligned}
F_{V^{\circ}}^{P}(x,a|z) &= \sum_{k=1}^{\infty} f_{V^{\circ}}^{(k)}(x,a) \,z^k
\AND\\
F_{V_n^{\circ}}^{P}(x,a|z) &= \sum_{k=1}^{\infty} f_{V_n^{\circ}}^{(k)}(x,a)\,z^k\,,\quad z \in \mathbb{C}.  
\end{aligned}
\end{equation}
For $\PQ \in \Ppi_+\,$, we extend the definition of \eqref{eq:specrad} by 
\begin{equation}\label{eq:radspec}
\lambda(\PQ_{V^a}) = 
%\begin{cases}
%                       \lambda(\PQ\,), &\text{if }\; \partial V = \emptyset \; (\, V^{\circ}=V\,),\\
                        \sup_{x,y \in V^a} \limsup_{n \to \infty} p_{V^a}^{(n)}(x,y)^{1/n}\,.
%                         &\text{if }\; \partial V \ne \emptyset \,.
%                        \end{cases}
\end{equation}
%To justify this, %the second part, 
Observe that deletion of the non-empty set $\partial_a V$ leaves ad most countably 
many connected components $C_i$ of our graph. Then each $\PQ_{C_i}$ is an irreducible, 
substochastic matrix, so that by old
and well-known results on infinite, non-negative matrices (see {\sc Seneta}~\cite{Seneta}),
$$
\lambda(\PQ_{C_i}) = \limsup_{n \to \infty} p_{C_i}^{(n)}(x,y)^{1/n} 
$$
is independent of $x,y \in C_i\,$, while $p_{V^a}^{(n)}(x,y) = 0$ when $x$ and $y$ 
do not belong to the same component. This means that $\lambda(\PQ_{V^a}) = \sup_i \lambda(\PQ_{C_i})$.
Recalling that $V^a = V^{\circ} \setminus \{ a \}$,
we thus get the following also in the infinite case.

\begin{lemma}\label{lem:FG-converge} If $\PQ \in \Ppi_+$ then the power series 
defining $G_{V^a}^{\PQ}(x,a|z)$ and thus also $F_{V^{\circ}}^{\PQ}(x,a|z)$,
as well as $F_{V_n^{\circ}}^{\PQ}(x,a|z)$,
converge absolutely whenever $|z| < 1/\lambda(\PQ_{\,V^{a}})$.
\end{lemma}

We now have the following comparison result for convergence of the respective power series.

\begin{theorem}\label{thm:conv}
Let $s \in \mathbb{H}_r$, $t > 0$ and $z \in \C$. If 
$$
|z| < \dfrac{1}{r_{s,t}\,\lambda(\PQ_{t|V^a})}
\quad \text{or} \quad |z| < \dfrac{1}{r_s\,\lambda(\,\wt{\!\!\PQ}_{s|V^a})}
\quad \text{or} \quad |z| < \dfrac{1}{r_s\,\lambda(\,\wc{\!\!\PQ}_{s|V^a})}
$$
then the power series of \eqref{eq:FFn} converge absolutely for $P = P_s\,$. Furthermore
$$
\lim_{n \to \infty} F_{V_n^{\circ}}^{P_s}(x,a|z) = F_{V^{\circ}}^{P_s}(x,a|z) \quad \text{for all }\;
x \in V^a\,.
$$
This also holds when $\partial V = \emptyset,$ i.e., $V^{\circ} = V$ and $V^a = V \setminus\{a\}$.
\end{theorem}

\begin{proof}  
A \emph{walk} in $(V,E)$ is a sequence $\ww = (x_0\,,x_1\,,\dots, x_k)$ of vertices such that
$[x_{i-1}\,,x_i] \in E$ for all $i$.  Its length is $k$, and for $z \in\C$, its \emph{$z$-weight}
with respect to $P \in \Ppi$ is
$$
W^P(\ww|z) = \prod_{i=1}^k \bigl(p(x_{i-1},x_i)\,z\bigr).  
$$
We also admit $k=0$, in which case the walk consists of a single vertex, and its
weight is defined as $1$.
If $\WW$ is a set of walks, then 
$$
W^P(\WW|z) = \sum_{\ww \in \WW} W^P(\ww|z).
$$
When $\WW$ is infinite, we require absolute convergence.
For any subset $U$ of $V$, and $x,y \in U$, let $\WW_U(x,y)$ be the set 
of all walks within $U$ which start at $x$ and end at $y$, and 
$\WWs_U(x,y)$ the set of those walks which meet $y$ only at their
endpoint. Finally,  the superscript $^{(k)}$ refers to the respective walks
of length $k$. Note that $\WW_U^{(k)}(x,y)$ is finite.
Then, referring to \eqref{eq:fk}, for $x, y \in V^a$ we have
$$
p_{V^a}^{(k)}(x,y) = W^P\bigl(\WW_{V^a}^{(k)}(x,y)|z\bigr)
\AND
f_{V^a}^{(k)}(x,a) = W^P\bigl(\WWs_{V^{\circ}}^{(k)}(x,a)|z\bigr).
$$
The analogous identities hold when whe replace $V$ by $V_n\,$. If $|z|$ is
sufficiently small to yield absolute convergence, we get
$$
G_{V^a}^P(x,y) = W^P\bigl(\WW_{V^a}(x,y)|z\bigr)
\AND 
F_{V^{\circ}}^P(x,a)= W^P\bigl(\WWs_{V^{\circ}}(x,a)|z\bigr),
$$ 
and with the same $z$, we may again replace $V$ by $V_n\,$.

Now suppose that $\PQ \in \Ppi_+$. Then ``sufficiently small'' just means
that $|z| < 1 / \lambda(\PQ_{V^a})$. 
We can apply this to $P_s$ with $s \in \mathbb{H}_r\,$, and to $\PQ = \wt{\!\PQ}_s$ or
$\PQ = \wc{\!\PQ}_s$
or $\PQ = \PQ_t$ with $t > 0$. Then 
$|p_s(x,y)| \le r\, \pq(x,y)$ with $r= r_{s,t}$ or $r = r_s\,$, respectively.
Then
$$
\bigl|W^{P_s}\bigl(\WW_{V^a}^{(k)}(x,y)|z\bigr)\bigr| \le  W^{\PQ}\bigl(\WW_{V^a}^{(k)}(x,y)\big|\, r|z| \bigr)
$$
for all $x,y \in V^a \supset V_n^a$ and all $k \ge 0$. When, as assumed, $r|z| < 1/\lambda(\,\PQ_{V^a})$,
we get that both  power series 
$F_{V^{\circ}}^{P_s}(x,a|z)$ and $F_{V_n^{\circ}}^{P_s}(x,a|z)$ are dominated in 
element-wise absolute value by 
$$
F_{V^{\circ}}^{\PQ}\bigl(x,a\big|\,r|z|\bigr) < \infty\,.
$$
The use of weights of walks serves in particular to verify the second statement of the theorem: 
for $r|z| < 1/\lambda(\,\PQ_{V^a})$,
absolute convergence allows us two estimate
$$
\begin{aligned}
\Bigl| F_{V^{\circ}}^{P_s}(x,a|z) - F_{V_n^{\circ}}^{P_s}(x,a|z) \Bigr| &=
\Bigl| W^{P_s}\Bigl(\WWs_{V^{\circ}}(x,a) \setminus \WWs_{V_n^{\circ}}(x,a)\,\big|\, z\Bigr) \Bigr|\\
&\le W^{\PQ}\Bigl(\WWs_{V^{\circ}}(x,a) \setminus \WWs_{V_n^{\circ}}(x,a)\,\big|\,r|z|\Bigr)\\
&= F_{V^{\circ}}^{\PQ}\bigl(x,a\big|\,r|z|\bigr) - F_{V_n^{\circ}}^{\PQ}\bigl(x,a\big|\,r|z|\bigr).
\end{aligned}
$$
By monotone convergence, the last difference tends to $0$. 
\end{proof}

We would like to apply the last theorem in particular to the unsrestricted transient case 
($\partial V = \emptyset$) with $z = 1$. When $s$ is non-real,
this requires that the stochastic comparison matrix $\PQ \in \Ppi_+\,$ satisfies $\lambda(\PQ\,) < 1$.
This is independent of the specific choice of $\PQ\,$ by Proposition \ref{pro:transrec}, and then 
Theorem \ref{thm:conv} applies when $|\Im s|$ is sufficiently small. Compare this with 
Example \ref{ex:ample}. For general $s \in \mathbb{H}_r\,$, let
$$
 \lambda(P_s)_{xy} = \limsup_{n \to \infty}\, |p_s^{(n)}(x,y)|^{1/n}\,. 
$$
Then $1/\lambda(P_s)_{xy}$ is the radius of convergence of the power series $G^{P_s}(x,y|z)$,
defined as in \eqref{eq:Gxyz}. However, contrary to the stochastic case, we do not see a general
argument that this should be independent of $x$ and $y$. Let us call
\begin{equation}\label{eq:lam}
 \lambda(P_s) = \sup \{ \lambda(P_s)_{xy} : x,y \in V\}
\end{equation}
the \emph{spectral radius} of $P_s\,$. In the stochastic case,  this is 
indeed the spectral radius (norm) as a self-adjoint operator, see \eqref{eq:m}.
For general $s \in \mathbb{H}_r\,$, we are not aware of an analogous interpretation.

Another question is the following. For stochastic $\PQ \in \Ppi_+\,$, the function 
$z \mapsto G^{\PQ}(x,y|1/z)/z$
is the $(x,y)$-matrix element of the resolvent operator $(z\cdot I-\PQ\,)^{-1}$, so that it extends
analytically to $\C \setminus \textsf{spec}(\PQ\,)$, where $\textsf{spec}(\PQ\,)$ is the spectrum of
$\PQ$ as an operator as described via \eqref{eq:m}. Since 
$\textsf{spec}(\PQ\,) \subset [-\lambda(\PQ\,)\,,\lambda(\PQ\,)]\,$ is real, we get that $G^{\PQ}(x,y|z)$
extends as a holomorphic function from the disk $\{ z \in \C : |z| < 1/\lambda(\PQ\,) \}$
to all $z \in \C \setminus \R$ with $|z| \ge 1/\lambda(\PQ\,)$. Is there a similar 
property for general $P_s\,$?

These observations and questions are also valid when $\partial V \ne \emptyset$. The same is true
for the next identities which we state only for empty boundary. Recall that we have $F^{P_s}(a,a)=1$
for every $a \in V$.

\begin{lemma}\label{lem:identities} For every $s \in \mathbb{H}_r\,$, the following holds.\\[4pt]
\emph{(a)} $(V,\rho_s)$ is recurrent if and only if $F^{P_s}(x,a) = 1$ for some $a \in V$ and 
all $x \sim a$. In this case, $F^{P_s}(x,y) = 1$ for all $x,y \in V$.
\\[3pt]
\emph{(b)} In the transient case, for every $a \in V$,
$$
G^{P_s}(a,a) = \frac{1}{1 - U^{P_s}(a,a)}\,,\quad \text{where}\quad 
U^{P_s}(a,a) = \sum_{x \sim a} p_s(a,x)\,F^{P_s}(x,a)\,. 
$$
\emph{(c)} For all $x,a \in V$ with $x \ne a$ (not necessarily neighbours)
$$
F^{P_s}(x,a) = \sum_{y} p_s(x,y)\, F^{P_s}(y,a).
$$
\emph{(d)} If $y$ is a cut vertex between $x$ and $a$ (i.e., every path from 
$x$ to $a$ passes through $y$) then
$$
F^{P_s}(x,a) = F^{P_s}(x,y)\,F^{P_s}(y,a).
$$
\end{lemma}

All these identities hold for $s > 0$, see e.g. \cite[\S 1.D]{Woess09}, and extend 
to complex $s \in \mathbb{H}_r$ by analytic extension, compare with the proof of
Theorem \ref{thm:converge}. They also hold for the generating functions
$G^{P_s}(a,a|z)$ and $F^{P_s}(x,y|z)$ with the adaptations
$$
\begin{aligned}
U^{P_s}(a,a|z) &= \sum_{x \sim a} p_s(a,x)z\,F^{P_s}(x,a|z) \AND\\ 
F^{P_s}(x,a|z) &= \sum_{y} p_s(x,y)z\, F^{P_s}(x,a|z)
\end{aligned}
$$
as long as $|z|< 1/\lambda(P_s)$, but it is not clear to us how to
bridge the gap between these values of $z$ and the value $1$ corresponding to
the statements of Lemma \ref{lem:identities}.

\section{Trees and free groups}\label{sec:trees}

In this section we concentrate on the infinite, transient case in absence of a 
finite set of grounded vertices. For $P_s \in \Ppi$, we shall write 
$G^s=G^{P_s}$ and $F^s=F^{P_s}$ for the associated kernels. The fact that we
have these kernels and that their matrix elements are holomorphic functions of
$s \in \mathbb{H}_r$ allows us to transport a variety of methods and results
from the stochastic case to this complex-weighted one. Here, we present some
example classes of this kind.

\medskip

\textbf{A. Trees and harmonic functions}

We assume that $V=T$ is (the vertex set of) an infinite, locally finite tree,
i.e., a connected graph wihout closed walks whose vertices are all distinct.
We assume that each vertex has at least two neighbours. We also assume that
our complex weights $\rho_s(x,y)$ are such that $P_s$ is transient for some
($\!\!\iff$ all) $s \in \mathbb{H}_r\,$. Taking up the definition of \S \ref{section:finiteGreen},
a function $h : T \to \C$ is called  \emph{harmonic} on $T$ if for all $x \in T$
$$
P_s h = h\,,\quad \text{where}\quad P_sh(x) = \sum_{y: y\sim x} p_s(x,y)h(y)\,.
$$
In this sub-section, we shall explain that every harmonic function has a 
Poisson-type boundary integral representation. 

We start by recalling the \emph{boundary at infinity} of the tree. First of all,  
for any pair of vertices $x,y$ there is a unique geodesic 
path $\pi(x,y)=[x=x_0\,,x_1\,,\dots, x_n=y]$ in $T$
from $x$ to $y$. A \emph{geodesic ray} is
a sequence $\pi = [x_0\,,x_1\,,x_2\,,\dots]$ of distinct vertices such that $x_k \sim x_{k-1}$
for all $k$. Two rays are called equivalent if (as sets) their symmetric difference is
finite, that is, they differ at most for finitely many initial vertices. An equivalence
class of rays is an \emph{end} of $T$. It represents a way (direction) of going to infinity
in $T$. 
The set of all ends $\partial^{\infty}T$ is the boundary at infinity of $T$.
For every $x \in T$ and $\xi \in \partial^{\infty}T$, there is a unique ray $\pi(x,\xi)$
with initial vertex $x$ which represents $\xi$.
We get the compact metric space $\widehat T = T \cup \partial^{\infty}T$ as follows.
We fix a ``root'' vertex $o$. The length $|x|$ of a vertex $x \in T$ is its graph distance
from $o$, that is, the number of edges of $\pi(o,x)$. For distinct $\xi, \eta \in \widehat{T}$, their 
\emph{confluent} $\xi \wedge \eta$ is the last common vertex on the geodesics
$\pi(o,\xi)$ and $\pi(o,\eta)$. Then
$$
\theta(\xi,\eta) = \begin{cases} 0\,&\text{if }\; \xi = \eta\,,\\
                                 2^{-|\xi \wedge \eta|}\,,&\text{if }\; \xi \ne \eta
                   \end{cases}
$$
defines an (ultra)metric on $\widehat{T}$, and $T$ becomes a discrete, dense subset of
the compact space $\widehat{T}$. A basis of the toppology on $\partial^{\infty} T$ 
is given by all \emph{boundary arcs}
$$
\partial^{\infty} T_x = \{ \xi \in \partial^{\infty} T: x \in \pi(o,\xi) \}\,,\quad x \in T\,.
$$
Each boundary arc is open-compact. A \emph{successor} of a vertex $x \in T$ is a neighbour
$y$ of $x$ such that $|y|=|x|+1$, and then we call $x=y^-$ the \emph{predecessor} of $y$.
We have 
$$
\partial^{\infty} T_x = \bigcup_{y^-=x} \partial^{\infty} T_y\,,
$$
a disjoint union. 

\begin{definition}\label{def:distribution} 
A \emph{distribution} on $\partial^{\infty} T$ is a finitely additive complex measure 
$\nu$ on $\mathcal{F} = \{ \partial^{\infty} T_x : x \in T \}$, that is,
$$
\nu(\partial^{\infty} T_x) = \sum_{y^-=x} \nu(\partial^{\infty} T_y)
$$
for all $x \in T$.
\end{definition}

\begin{remark}\label{rmk:cohen} In \cite{Cohen}, the following is proved. 
 A distribution $\nu$ on $\partial^{\infty} T$ extends to a complex Borel measure
 on the compact space $\partial^{\infty} T$ if an only if for any family 
 of mutually disjoint boundary arcs $\partial^{\infty} T_{x_n}$, $n \in \N$,
 one has
 $$
 \sum_n \bigl|\nu(\partial^{\infty} T_{x_n})\bigr| < \infty.
 $$
\end{remark}

If $\varphi$ is a locally constant function on $\partial^{\infty}T$ then it can be
written as a linear combination of indicator functions of boundary arcs,
$$
\varphi = \sum_{j=1}^k c_k\,\uno_{\partial^{\infty} T_{x_j}}\,,\quad c_j \in \C\,,
$$
and in this case, the arcs can be forced to be pairwise disjoint. For a distribution
$\nu$ as in Definition \ref{def:distribution}, we then set 
\begin{equation}\label{eq:integ}
\int_{\partial^{\infty} T} \varphi\, d\nu = \sum_{j=1}^k c_k\,\nu(\partial^{\infty} T_{x_j})\,.
\end{equation}
As a matter of fact, via this definition, the linear space of all distributions is the dual
of the space of all locally constant functions on $\partial^{\infty}T$, compare with
\cite{Cohen}.

\smallskip

In addition to Lemma \ref{lem:identities}, we now shall need the following, which is specific
to trees.

\begin{lemma}\label{lem:tree-id}
Suppose that $(T, \rho_s)$ is transient.  
Then for every $s \in \mathbb{H}_r\,$ and every pair of neighbours $x,y \in T\,$, 
$$
p_s(x,y) \bigl( 1 - F^{P_s}(x,y)F^{P_s}(y,x)\bigr) = F^{P_s}(x,y) \bigl( 1 - U^{P_s}(y,x) \bigr).
$$
In particular,
$$
F^{P_s}(x,y)F^{P_s}(y,x) \ne 1 \AND F^{P_s}(x,y) \ne 0\,.
$$
\end{lemma}

\begin{proof} For real $s > 0$, the identity is derived in \cite[(9.35)]{Woess09}.
Once more, it must hold for all  $s \in \mathbb{H}_r$ by analytic extension.
\end{proof}

Note that for arbitrary $x,y \in T$, if $[x=x_0\,,x_1\,,\dots, x_n=y]$ is the geodesic
path connecting the two, then the tree structure and Lemma \ref{lem:identities}(c) imply that
\begin{equation}\label{eq:fsplit}
F^{P_s}(x,y) = F^{P_s}(x,x_1)F^{P_s}(x_1,x_2) \cdots F^{P_s}(x_{n-1},y) \ne 0.
\end{equation}

\begin{corollary}\label{cor:gne0}
In the transient case, $G^{P_s}(x,y) \ne 0$ for all  $x,y \in T$.
\end{corollary}

At this point, we can define the \emph{Martin kernel} as in the stochastic case:
$$
K^{P_s}(x,\xi) = \lim_{y \to \xi} \frac{G^{P_s}(x,y)}{G^{P_s}(o,y)} 
= \frac{F^{P_s}(x,x \wedge \xi)}{F^{P_s}(o,x \wedge \xi)}\,,
\quad (x,\xi) \in T \times \partial^{\infty}T\,.
$$
The second identity follows from \eqref{eq:fsplit}. Note that for any fixed $x \in T$,
the function $\xi \mapsto K(x,\xi)$ is locally constant on $\partial^{\infty}T\,$.
We now get the following extension of a result which is well-known in the stochastic case.

\begin{theorem}\label{thm:martin}
Harmonic functions are in one-to-one correspondence with distributions on $\partial^{\infty}T\,$:
for every harmonic function $h$ on $T$ with respect to $P_s$ ($s \in \mathbb{H}_r$) , 
there is a unique distribution $\nu^h$ on $\partial^{\infty} T$ such that
$$
h(x) = \int_{\partial^{\infty} T}  K^{P_s}(x,\xi)\, d\nu^h(\xi) \quad \text{for all }\; x \in T\,.
$$
The distribution is given by 
$$
\nu^h(\partial^{\infty} T_x) 
= F^{P_s}(o,x)\, \frac{h(x)- F^{P_s}(x,x^-)h(x^-)}{1-F^{P_s}(x,x^-)F^{P_s}(x^-,x)}\,,\quad 
x \in T \setminus \{o\}\,,
$$
and $\nu^h(\partial^{\infty} T)=h(o)$.
\end{theorem}

The proof is exactly as in \cite[Theorem 9.36]{Woess09}. It goes back to \cite{Cartier}.

More generally, for $\lambda \in \C$, a function $h: T \to \C$ is called \emph{$\lambda$-harmonic}
with respect to $(T, \rho_s)$ if $P_sh = \lambda\cdot h$. For suitable values of $\lambda$,
the above extends to $\lambda$-harmonic functions. Namely, if for $t > 0$
\begin{equation}\label{eq:range}
|\lambda| > r_{s,t}\,\lambda(\PQ_{t|V^a}) 
\quad \text{or} \quad |\lambda| > r_s\,\lambda(\wt{\!\PQ}_{s|V^a})
\quad \text{or} \quad |\lambda| > r_s\,\lambda(\wc{\!\PQ}_{s|V^a})
\end{equation}
then we can use comparison and work with $G^{P_s}(x,y|1/\lambda)$ and $F^{P_s}(x,y|1/\lambda)$.
The associated Martin kernel is then
$$
K^{P_s}(x,\xi|\lambda) %= \lim_{y \to \xi} \frac{G^{P_s}(x,y)}{G^{P_s}(o,y)} 
= \frac{F^{P_s}(x,x \wedge \xi|1/\lambda)}{F^{P_s}(o,x \wedge \xi|1/\lambda)}\,,
\quad (x,\xi) \in T \times \partial^{\infty}T\,.
$$
In this case, the arguments of \cite[(9.35)]{Woess09} that lead to Lemma \ref{lem:tree-id}
can be applied directly via ``path composition'' as in that reference, and 
one gets 
$$
p_s(x,y) \bigl( 1 - F^{P_s}(x,y|1/\lambda)F^{P_s}(y,x1/\lambda)\bigr) 
= \lambda\, F^{P_s}(x,y|1/\lambda) \bigl( 1 - U^{P_s}(y,x|1/\lambda) \bigr).
$$
Thereafter, everything works as in \cite{PicardelloWoess} (with a little care
concerning the slightly different notation), and one gets the 
analogue of Theorem \ref{thm:martin} with $\nu^h$ as in that theorem, replacing the 
appearing terms $F^{P_s}(\cdot,\cdot)$ with $F^{P_s}(\cdot,\cdot|1/\lambda)$.
Following the methods of \cite{PicardelloWoess}, one also gets boundary
integral representations of \emph{$\lambda$-polyharmonic functions} for complex $\lambda$
in the range of \eqref{eq:range}.

However, in general $\lambda = 1$ does not belong to that range, unless the stochastic
operators have spectral radius strictly $< 1$ and $|s|/\Re s$ is sufficiently close to $1$.
One of the future issues is to understand if and how the gap between $\lambda =1$ and 
$\lambda$ in the range of \eqref{eq:range} can be bridged. The (finite) Example \ref{ex:ample}
indicates that this will not always be possible.

\medskip

\textbf{B. Free groups}

We consider the case when $V = \Gamma$ is a finitely generated group and 
$A$ is a finite, symmetric set of generators of $\Gamma$ which does not
contain the group identity $e$. The Cayley graph of $\Gamma$ has vertex
set $\Gamma$, and two vertices $x, y$ are neighbours if and only if $x^{-1}y \in A$.
Then it is natural to require that the edge admittances \eqref{eq:admittance}
satisfy
$$
\rho_s(x,y) = \rho_s(e, x^{-1}y)\,,
$$
so that $a \mapsto \rho_s(e,a)$ is a non-zero, symmetric function $A \to \C$.
We then have $\rho_s(x) = \rho_s(e) = \sum_{a \in A} \rho_s(e,a)$ 
for the total admittance at any group element (vertex) $x$. We get that
$$
p_s(x,y) = \mu_s(x^{-1}y)\,, \quad \text{where}\quad \mu_s(x) = \frac{\rho_s(e,x)}{\rho_s(e)}
$$
is a symmetric, complex measure supported by $A$ with total sum $1$. 
The transition operator $P_s$ is then the right convolution operator by $\mu_s\,$, and in 
the subsequent notation, we shall alsways refer to $\mu_s$ in the place of $P_s\,$. It is 
natural to consider the action on $\ell^2(\Gamma)$, the 
Hilbert space of all square summable
complex functions on $\Gamma$. The operator is symmetric, but not self-adjoint unless $s > 0$.
We are interested in its norm $\|\mu_s\|_{\ell^2}$ and its operator spectral radius
$$
\lambda_{\ell^2}(\mu_s) = \lim_{n \to \infty} \|\mu_s^{(n)}\|_{\ell^2}^{1/n}\,, 
$$
where $\mu_s^{(n)}$ is the $n^{\text{th}}$ convolution power of $\mu_s$. 
For  the ``spectral radius'' $\lambda(\mu_s)=\lambda(P_s)$ defined 
in \eqref{eq:lam}, we have 
$$
\lambda(\mu_s) \le \lambda_{\ell^2}(\mu_s) \le \|\mu_s\|_{\ell^2}\,,\quad s \in \mathbb{H}_r\,.
$$
When $s > 0$, the three numbers coincide, with $\lambda(\mu_s)$ being  
the associated Markov chain spectral radius \eqref{eq:specrad}, and $\mu_s$ is of course
a probability measure on $\Gamma$.

For any $s \in \mathbb{H}_r$, when $|z| < 1/ \lambda(\mu_s)$, we get convergence
of the power series $G^{\mu_s}(x,y|z)$. This holds in particular, when 
$|z| < 1/ \|\mu_s\|_{\ell^2}\,$.

\smallskip

We now consider the case when $\Gamma$ is a \emph{free group} with free 
generators $a_1\,,\dots, a_k$ ($k \ge 2$). We set $a_{-j} = a_j^{-1}$ 
and $A = \{ a_{\pm j} : j = 1, \dots, k\}$ for our symmetric generating set.
Recall that $\Gamma$ consists of all \emph{reduced words}
$$
x = a_{j_1}a_{j_2} \cdots a_{j_n}\,,\quad n \ge 0\,,\; 
j_l \in J = \{\pm 1, \dots, \pm k\}, j_l \ne -j_{l-1}\,. 
$$
When $n =0$, this is the empty word, which stands for the group identity $e$.
The group operation is concatenation of words followed by reduction, i.e.,
cancellation of successive ``letter'' pairs $a_ja_{-j}\,$, $j \in J$.

The Cayley graph of $\Gamma$ with respect to $A$ is the regular tree where 
each vertex has $2k$ neighbours.  It is very well known since \cite{Kesten} that 
$\lambda(\mu_s) < 1$ in the stochastic case $s > 0$, and we have transience. 
In particular, the results of the preceding sub-section apply here.
The following important
result is due to \cite{AkOs}; for a simple ``random walk'' proof, see
\cite{Woess86}.

\begin{proposition}\label{pro:AkOs}
For any $s \in \mathbb{H}_r\,$,
$$
\|\mu_s\|_{\ell^2} = 
2 \cdot \min \biggl\{ t + \sum_{j=1}^k \Bigl( \sqrt{t^2 + |\mu_s(a_j)|^2} - t\Bigr) : t \ge 0 \biggr\}.
$$
\end{proposition}

In particular, the norm is the same as for  $|\mu|_s$, where $|\mu|_s(x) = |\mu_s(x)|$. The latter
is in general not a probability measure, its total mass is $\ge 1$. Thus, we may have 
$\|\mu_s\|_{\ell^2} \ge 1$ when $s$ is complex. We have by Proposition \ref{pro:le}
\begin{equation}\label{eq:ineq2}
|\mu|_s(\Gamma) = 2\sum_{j=1}^k |\mu_s(a_j)| 
= \frac{\sum_{j}|\rho_s(e,a_j)|}{\Bigl|\sum_j \rho_s(e,a_j)\Bigr|} \le \frac{|s|}{\Re s},
\end{equation}
and since $\frac{1}{|\mu|_s(\Gamma)}|\mu|_s$ is a probability measure, its operator norm
(= spectral radius) is $< 1$. The stochastic transition operator induced by this probability
measure is $\wc{\!\PQ}_s\,$. Again, if $\Re s/|s|$ is sufficiently close to $1$, we can 
use the comparison method described in the previous sections, including the Green kernel 
at $z = 1$. As a matter of fact, this applies to any non-amenable group, but
here we have a specific formula for the norm.

\begin{example}\label{exa:mple2}
We suppose that $\rho_s(e,a_j)=\rho_s(e,a_{-j}) \in \{ 1, s, 1/s \}$ and that $s = e^{i\alpha}$ 
with $|\alpha| < \pi/2$.

Note that then admittance $=1$ means $R=1,\ L=D=0$, admittance $=s$ means
$D=1,\ L=R=0$, and admittance $=1/s$ means $L=1,\ R=D=0$. Let  
$$
\begin{aligned}
&|\{ j \in \{1, \dots, k\} : \rho_s(e,a_j)=s  \}| = l_1\,,\\
&|\{ j \in \{1, \dots, k\} : \rho_s(e,a_j)=1/s  \}| = l_2\,,\AND\\
&|\{ j \in \{1, \dots, k\} : \rho_s(e,a_j)=1  \}| = l_3 = k - l_1-l_2\,.
\end{aligned}
$$
Then $\frac{1}{|\mu|_s(\Gamma)}|\mu|_s$ is equidistribution on $A$, and it is very well known
that the norm of the associated convolution operator, i.e., the spectral radius of simple
random walk is  $\sqrt{2k-1}/k$. Consequently,
$$
\|\mu_s\|_{\ell^2} = \frac{\sqrt{2k-1}}{|l_1\, s + l_2/s + l_3|}\,. 
$$
If $l_3 > l_1 + l_2$ then 
$$
\|\mu_s\|_{\ell^2} \le \frac{\sqrt{2k-1}}{k - 2(l_1+l_2)}\,, 
$$
and if $l_1+l_2$ is fixed, then this will be $< 1$ for $k$ sufficiently large,
so that the Green kernel $G^{\mu_s}(x,y|z)$ is defined via the corresponding power series 
for complex $z$ in an open disk around the origin that contains $z=1$. 

The same is true when $k$ is small and the angle $\alpha$ is sufficiently close to $0$.
For example, when $k=2$ and $l_1=l_2=1$ then $\|\mu_s\|_{\ell^2} = \sqrt{3}/(2\cos \alpha)$
which is $<1$ when $|\alpha| < \pi/6$. 

Of course, the general estimate \eqref{eq:ineq2} yields a smaller range of angles 
$\alpha$ for which one obtains $\|\mu_s\|_{\ell^2} < 1$, like in the finite network
of Example \ref{ex:ample}.

In all those cases, the power series representation of the Green kernel in a neighbourhood of $z=1$ 
allows to  derive a variety of further results, such as the study of polyharmonic 
functions as in \cite{PicardelloWoess}.
\end{example}

\end{document}